\documentclass{article}
\usepackage{fullpage}

\usepackage{color}
\usepackage{epsfig}
\usepackage{amsthm}
\usepackage{amsmath}
\usepackage{enumerate}
\usepackage{setspace}
\usepackage{wasysym}
\newcommand{\bea}{\begin{eqnarray}}
\newcommand{\eea}{\end{eqnarray}}

\newtheorem{theorem}{Theorem}
\newtheorem{lemma}[theorem]{Lemma}
\newtheorem{obs}[theorem]{Observation}
\newtheorem{cor}[theorem]{Corollary}

\newcommand{\ora}{\overrightarrow}
\newcommand{\ola}{\overleftarrow}

\newcommand{\tw}{w}

\begin{document}

\title{Improving Robustness of Next-Hop Routing\thanks{This material is based upon work supported by the National Science
  Foundation, under Grant No.\ CCF-0964037.  Sean Kennedy is partially supported by a postdoctoral fellowship from the Natrual Sciences and Engineering Research Council of Canada.}}

\author{
Glencora Borradaile, W.~Sean~Kennedy, Gordon Wilfong, Lisa Zhang
}

% \institute{G. Borradaile \at 
% School of Electrical Engineering and Computer Science \\Oregon State University\\ \email{glencora@eecs.oregonstate.edu}
% \and
% W.~S.~Kennedy, G. Wilfong, and L. Zhang \at 
% Bell Labs\\ Murray Hill, NJ\\ \email{\{kennedy,gtw,ylz\}@research.bell-labs.com}}

\maketitle

\pagenumbering{arabic}

\begin{abstract}
A weakness of next-hop routing is that following a link or router failure there may be no routes between some source-destination pairs, or packets may get stuck in a routing loop as the protocol operates to establish new routes.
In this article, we address these weaknesses by describing mechanisms to choose alternate next hops.

Our first contribution is to model the scenario as the following {\sc tree
augmentation} problem.  Consider a mixed graph where some edges are
directed and some undirected.  The directed edges form a spanning tree
pointing towards the common destination node. Each directed edge
represents the unique next hop in the routing protocol.  Our goal is
to direct the undirected edges so that the resulting graph remains
acyclic and the number of nodes with outdegree two or more is
maximized.  These nodes represent those with alternative next hops
in their routing paths.

We show that {\sc tree augmentation} is NP-hard in general and present
a simple $\frac{1}{2}$-approximation algorithm.  We also study 3 special
cases.  We give exact polynomial-time algorithms for when the input spanning tree consists of exactly 2 directed
paths or when the input graph has bounded treewidth.  For planar graphs, we present a
polynomial-time approximation scheme when the input tree is a
breadth-first search tree. To the best of our knowledge, {\sc tree
  augmentation} has not been previously studied.

%\keywords{graph \and orientation \and acyclic \and tree}
\end{abstract}

\section{Introduction}

In an Internet Protocol (IP) network the header of each packet
contains the intended destination for that packet.  Each router in the
network has a {\em forwarding table} in which, for each destination, there is a corresponding entry stating which port through which to send out packets.  This type of routing is called {\em next-hop routing}.
Since the routings for different destinations are independent, without
loss of generality, we will herein talk about routing to one particular
destination $d$.  A number of different routing protocols are
available for populating routing tables with  destination-port
pairs.  These protocols are called {\em interior gateway protocols}
(IGPs); see References~\cite{rfc:1058,rfc:2178,rfc:1142} for detailed operations of various of these IGPs.

For our purposes, an IP network is modeled as a graph, each node
representing a router and each edge representing a link between
routers.  The entry in the routing table for destination $d$ at router
$v$ is a pair $(d,vw)$ where $vw$ is a directed edge to a neighboring
router $w$ of $v$.  The goal of an IGP is to design routing tables so
that the union of these edges is a tree directed towards $d$; we call such a table {\em valid}.  The unique directed path from a router $v$ to destination $d$ in this tree implies a {\em route} for a message from $v$ to $d$.

A weakness of next-hop routing is that following a link or router
failure there may be no route from some source router to $d$.  While a
routing protocol operates to establish new routes (by way of new
router table entries), packets may get stuck in a routing loop.  The {\em speed of
  convergence}, i.e.~the time required to recover from such a failure in the
network topology, is of fundamental importance in modern IGPs.
Indeed, real-time applications such as voice-over IP require quick
failure recovery.  

There have been many different methods proposed to
recover from failures; see References~\cite{GSBCSHT12} and~\cite{RI07} for surveys.  
We address a requirement for one such method, {\em Permutation Routing}~\cite{VLK12}.  In permutation routing, there may be multiple entries for destination $d$ at router $v$.  The route used by a message will use the first entry that corresponds to a link or neighbor that is not currently in failure.  To avoid routing loops, the union of the edges given by these entries should be a directed acyclic graph with $d$ as the only sink.   So long as all the neighbors or edges listed in a routers table do not {\em all} fail, the tables will imply at least one route to $d$.  In order to implement permutation routing, we start with a network of communication connections between routers and single-entry routing tables corresponding to a tree, often a shortest-path or breadth-first-search (BFS) tree of the graph.  We use the edges of the network to add entries to the router tables, thus adding the resiliency required for permutation routing.  As we show in this paper, it is the choice of how to add these entries that proves to be challenging.

In this paper, we study a graph theoretic problem which exactly models
this issue. Our contribution
is two-fold. From the networking perspective, we are the first to
offer a rigorous treatment of this approach for improving next-hop
routing robustness. From the graph theoretic perspective, we introduce a
crisply-stated graph problem that appears to be novel.

\subsection{The {\sc tree augmentation} problem}

We formally define the problem in graph terms.  The input is a mixed
graph $G=(V, E\cup \overrightarrow T)$, where $V$ is the set of nodes,
$E$ is a set of undirected edges and $\overrightarrow T$ is a set of
directed edges that form a directed spanning tree pointing towards the
destination node $d$.  (Such a rooted, directed spanning tree is in
fact an arborescence; we use the term tree for convenience.)  The
edges of $\overrightarrow T$ give the unique next-hop for the initial
routing paths.  Our goal is to find an orientation of the undirected
edges $E$ so that the resulting directed graph is acyclic and {\em
  maximizes} the number of nodes with outdegree at least two. Those
nodes that, in the resulting orientation, have outdegree at least two
represent those routers with an alternate next hop when a failure occurs. We
refer to this problem as {\sc tree augmentation}.

Herein, an {\em edge} $uv$ refers to an undirected edge between nodes
$u$ and $v$ and $\overrightarrow{uv}$ and $\overleftarrow{uv}$ refer
to directed edges, or arcs, from $u$ to $v$ and $v$ to $u$, respectively. 
For $E' \subseteq E$ we use $\overrightarrow{E'}$ to denote the set
of arcs corresponding to some orientation of $E'$; the details of that
orientation will be given by context.  We say a node $u \in V$ is {\em
covered} by an orientation if there is an arc $\overrightarrow{uv} \in
\overrightarrow{E'}$.  Our objective is to maximize {\em the number of
covered nodes} over all possible orientations of subsets $E'$ of $E$
such that $(V, \overrightarrow{E'} \cup \overrightarrow T)$ is an {\em
acyclic} graph.  Note that the problem statement does not require
every edge in $E$ be directed in the solution. However, we will see
that it never hurts to direct every edge in $E$.

We also consider a weighted version of {\sc tree augmentation}, for
which some nodes may be viewed as more important to be covered. Here,
a positive weight $w(v)$ is given for each node $v$ and the objective
becomes maximizing the total weight of covered nodes.

\subsection{Organization}

In Section~\ref{sec:comp} we present two observations and give simple a
$\frac{1}{2}$-approximation for {\sc tree augmentation}.  In
Section~\ref{sec:hard} we prove the NP-hardness of {\sc tree augmentation} via a reduction from
the {\sc set cover} problem. The remainder of the paper gives positive
results for some interesting special cases.  In Section
\ref{sec:boundedtw}, we describe a polynomial-time algorithm for
graphs of bounded treewidth.  In Section \ref{sec:planar}, we describe
a polynomial-time approximation scheme for planar graphs when $T$ is a
BFS tree. In Section \ref{sec:twoarm}, we describe a polynomial-time
algorithm for the special case in which the spanning tree consists of
exactly two directed paths.  We point out
that all our algorithmic results generalize to the weighted case.

\section{Observations and a simple approximation}\label{sec:comp}

We first note that, algorithmically, the acyclicity constraint is what
makes this problem challenging.  Consider the connected components of
the graph $G = (V,E)$.  In each component we want to orient the edges
so that each node is covered.  If a connected component is a tree,
root this tree at an arbitrary node and orient every edge towards the
chosen root.  In this way, every node except the root is covered.
Further, it is not possible for every node in a tree to be covered.
If a connected component is not a tree, we begin with an arbitrary
spanning tree $S$.  There must exist one node incident to an edge not
in $S$.  We root the spanning tree $S$ at this node and orient every
edge in $S$ towards the root.  Since the root is incident to an edge
outside $S$, we orient this edge away from the root.  All other edges
outside $S$ can be oriented arbitrarily.  In this way, every node of
the connected component is covered.   Since all these operations are
no more difficult that depth-first search, we get:

\begin{obs}\label{obs:cyclic}
  Finding an orientation of $E$ to maximize the number of covered
  vertices can be done in linear time.
\end{obs}

However, our guiding application requires that we find an acyclic
orientation.  Suppose we have oriented $E' \subseteq E$ such that
$(V,\overrightarrow T\cup \overrightarrow{E'})$ is acyclic.  Consider
a topological ordering of the vertices of this acyclic graph and consider any
edge $uv \in E \setminus E'$ in which, without loss of generality,
$u$ is before $v$ in the topological ordering.  Then
$(V,\overrightarrow T\cup \overrightarrow{E'} \cup
\overrightarrow{uv})$ is also acyclic.  Augmenting $(V,\overrightarrow T\cup \overrightarrow{E'})$ in this way ensures that $u$ is
covered while not affecting the coverage of any other vertex.  Repeating this for every
edge of $E \setminus E'$ gives:

\begin{obs} \label{obs:complete} Given an orientation of a subset of
  $E'$ of the non-tree edges $E$, one can always orient the remaining
  edges $E \setminus E'$ while maintaining acyclicity and without decreasing 
  the objective.
\end{obs}

Consider the bipartition of $E$ into the set of {\em back edges} $B$
(edges $uv \in E$ such that $u$ is a descendent of $v$ in
$\overrightarrow T$) and the {\em cross edges} $C$ (edges $uv$ in $E$ such that $u$ and $v$  have no ancestor/descendent relationship).  Each edge of $B$ can
only be oriented in one way, from descendent to ancestor, without
introducing a cycle.  We let $\overrightarrow B$ denote this
orientation.  
Consider two orientations of $C$: let $\overrightarrow C$ be the
orientation such that each edge $uv \in C$ is oriented from low to
high pre-order (the order given by the first time a vertex is visited
by a depth-first search) and let $\overleftarrow C$ be the reverse of
this orientation (it is, in fact, the order such that each edge $uv
\in C$ is oriented from low to high {\em post}-order).

\begin{lemma}\label{lem:acyc}
  $(V,\overrightarrow T \cup \overrightarrow B \cup \overrightarrow C)$
and $(V,\overrightarrow T \cup \overrightarrow B \cup \overleftarrow
C)$ are acyclic.
\end{lemma}

\begin{proof}
  Trivially, $(V,\overrightarrow T \cup \overrightarrow B)$ is acyclic.

  For a contradiction, suppose there is a directed cycle $\overrightarrow D$ in $(V,\overrightarrow T \cup \overrightarrow B \cup \overrightarrow C)$.  Let $\overrightarrow D'$ be the cycle obtained from $\overrightarrow D$ by contracting the edges in $\overrightarrow D$ that are not in $\overrightarrow C$ (i.e. $\overrightarrow D' \subseteq \overrightarrow C$.  Let $u_1, u_2, \ldots, u_\ell$ be the vertices of $\overrightarrow D'$ in order, starting from an arbitrary vertex $u_1$ of $\overrightarrow D'$; we have that $\ell \ge 2$ for otherwise the edges of $\overrightarrow D$ could not have a consistent direction.

  By the definition of the orientation of $\overrightarrow{u_iu_{i+1}}$, $u_i$ must appear before $u_{i+1}$ in the pre-order used for each $i = 1,\ldots, \ell-1$.  Likewise, by the definition of the orientation of $\overrightarrow{u_\ell u_1}$, $u_\ell$ must appear before $u_{1}$ in the pre-order used, hence a contradiction.

  Likewise, reversing the direction of the edges in $\overrightarrow C$ results in an acyclic graph $(V,\overrightarrow T \cup \overrightarrow B \cup \overleftarrow
C)$. \qed
\end{proof}

We can view these orientations in the
following way: embed $\overrightarrow T$ in a non-crossing way with
the root at the top and all edges directed upward; consider a DFS
traversal that explores the branches from left to right; $\overrightarrow
C$ is the orientation in which all the edges of $C$ are oriented from
left to right in this embedding and $\overleftarrow C$ is the
orientation in which all the edges of $C$ are oriented from right to
left in this embedding. These orientations are illustrated in Figure~\ref{fig:tight}(c) and~(d), respectively.  We use this observation to design a simple $\frac{1}{2}$-approximation to the {\sc tree augmentation} problem; that is, the algorithm is guaranteed to return an orientation that covers at least a half of the number of vertices that can be covered in an optimal orientation.

 \begin{figure}[ht]
   %\centering
   \centerline{
	\input{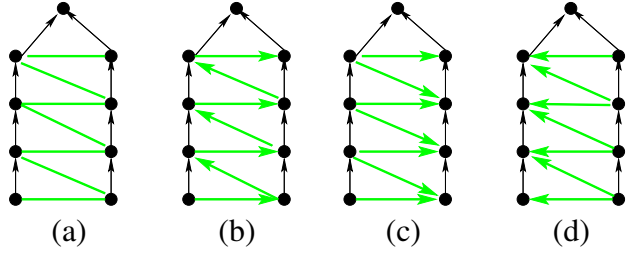}
	}
    \caption{The optimal solution (b) to the input problem (a) covers
      $n-2$ vertices whereas the
      left-to-right (c) and right-to-left (d) orientations cover
      $\frac{n-1}{2}$ vertices each.}
   \label{fig:tight}
 \end{figure}

\begin{theorem}\label{thm:2opt}
  The better of the two orientations $\overrightarrow B \cup
  \overrightarrow C$ and $\overrightarrow B \cup \overleftarrow C$ gives
  a $\frac{1}{2}$-approximation to the {\sc tree augmentation} problem.
\end{theorem}

\begin{proof}
  Let $V_B$ be the subset of vertices that are covered by
  $\overrightarrow B$.  Let $V_X$ be the subset of vertices that
  cannot be covered by any orientation (namely, those vertices that
  are not endpoints of edges in $E$).  The remaining vertices $V_C = V
  \setminus (V_B \cup V_X)$ are the endpoints of the edges in $C$.
  Each vertex in $V_C$ is covered either by $\overrightarrow C$ or $\overleftarrow C$ (or both). Let $x(\overrightarrow C) $ and $x (\overleftarrow C)$, be the number of vertices
  covered by $\overrightarrow C$ and  $\overleftarrow C$, respectively.  We have $x(\overrightarrow C) + x (\overleftarrow C) \ge |V_C|$. It follows that the
  number of vertices covered by the better of these two orientations
  has value:
  \begin{eqnarray*}
  \max \left\{|V_B|+x(\overrightarrow C),|V_B|+  x(\overleftarrow C)  \right\} &\geq&
  \frac{1}{2}\left(2|V_B|+x(\overrightarrow C) + x (\overleftarrow C) \right)
  \\&\ge& |V_B|+\frac{1}{2} |V_C| \ge\frac{1}{2} |V \setminus V_X|
  \end{eqnarray*}
  Since the maximum number of vertices that any orientation can cover
  is $|V \setminus V_X|$, the better of the two orientations
  $\overrightarrow B \cup \overrightarrow C$ and $\overrightarrow B
  \cup \overleftarrow C$ is a $\frac{1}{2}$-approximation to the {\sc tree
    augmentation} problem.  The example in Figure~\ref{fig:tight}~(b)
  illustrates that this analysis is asymptotically tight.
\qed
\end{proof} 

The $\frac{1}{2}$-approximation algorithm generalizes to the weighted case, 
for which each node $v$ carries a weight $w(v)$.
In the above proof we can simply replace the size of a node set
with the total weight from the set.
\begin{cor}
{\sc Weighted Tree Augmentation} admits a $\frac{1}{2}$-approximation.
\end{cor}

\section{{\sc Tree Augmentation} is NP-hard}
\label{sec:hard}

In this section, we prove the following.
\begin{theorem}
\label{thm:hard}
{\sc Tree Augmentation} is NP hard.
\end{theorem}
Our reduction is from the well-known NP-hard {\sc set cover} problem~\cite{GJ79}. 
An instance of {\sc set cover} consists of a set of elements $X=\{x_1,x_2,\ldots , x_n\}$, a
collection of subsets of $X$, ${\cal C}=\{S_1,S_2,\ldots , S_m\}$ and an
integer $0<k\le m$.  The {\sc set cover} problem asks: Is there a
subcollection ${\cal C}'\subseteq {\cal C}$ such that $|{\cal C}'|\le k$ and $\cup_{S_i\in
  {\cal C}'} S_i = X$?

We prove Theorem \ref{thm:hard} in three steps. 
First, we start by describing the gadget which models {\sc set cover} as an instance of {\sc tree augmentation}.
Second, we describe the intuition behind our result. 
Finally, we give the formal details of correctness.

\subsubsection*{The Reduction}  
We build an instance $G=(V,E\cup \overrightarrow T)$ of {\sc tree augmentation}
corresponding to an instance of {\sc set cover} as follows.  The construction is
illustrated in Figure~\ref{bad.fig}.

\paragraph{Vertices $V$:} The root vertex is $u$.  For each set $S_j$ ($j =
  1, \ldots, m$) we define 5 vertices $r_j, s_j, t_j, \ell_j, p_j$.
  We call the vertices $\{s_1, \ldots, s_m\}$ the {\em set vertices}.
  For each element $x_i$ ($i = 1, \ldots, n$) we
  define a set of $k+1$ {\em element vertices} $X_i = \{x_i^1, \ldots, x_i^{k+1}\}$.
\paragraph{Tree arcs} $\mathbf{\ora{T}}$: 
The tree $\mathbf{\ora{T}}$ consists of the following arcs.
  \begin{enumerate}
	  \item There is a directed path in $\ora{T}$ consisting of the arcs $\overrightarrow{p_ip_{i+1}}$
		  for $1\le i < m$, followed by the arc $\overrightarrow{p_mu}$.
		  We call this path the {\em collection path}.
	  \item For each set $S_j$, there is a directed path in $\ora{T}$ through the vertices
    $r_j,s_j,t_j,u$ in order.  We call these the {\em set paths}.
  \item For each element $x_i$, there is a directed path in $\ora{T}$ consisting of
	  the arcs $\overrightarrow{x_i^h x_i^{h+1}}$ for $1\le i<k+1$.
    We call these the {\em element paths}.
  \item The arc $\overrightarrow{x_n^{k+1}p_1}$ is in $\ora{T}$.
  \item For each $j = 1, ..., m$, $\ora{T}$ contains the arc $\overrightarrow{\ell_j x_1^1}.$
  \end{enumerate}
\paragraph{Non-tree edges $E$:}  For each set $S_j$, we connect every
  element vertex (corresponding to the elements in $S_j$) to the
  corresponding set vertex with non-tree edges $\{s_jx_i^1,\ldots,
  s_jx_i^{k+1}\ : \ x_i \in S_j\}$.  For each set $S_j$ there are 4
  additional non-tree edges: $\{\ell_ju, r_ju, \ell_jt_j, p_jr_j\}$.

\begin{figure}[htb]
\centerline{
\input{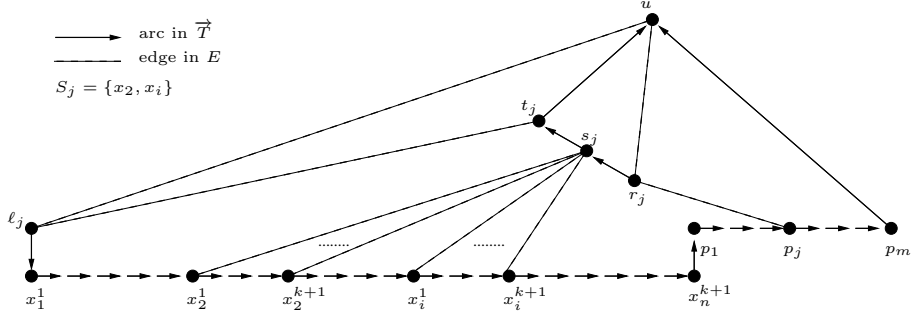}
}
\caption{NP-hardness construction.  The vertices $\ell_j, p_j, r_j,
  s_j, t_j$ (and adjacent edges) are only shown for one set $S_j$.}
\label{bad.fig}
\end{figure}

\subsubsection*{The Idea}
We start by noting that the orientation of $2m$ of the non-tree edges are forced in any acyclic orientation.
\begin{obs}\label{obs:forced}
In any feasible solution, i.e., any acyclic orientation, each edge $r_j u$ must
be oriented as $\ora{r_ju}$ since otherwise $ur_js_jt_j$ would form a directed cycle.
Similarly, in any feasible solution, each edge $\ell_ju$ must be oriented as $\ora{\ell_ju}$ since otherwise $u\ell_jt_j$ would form a directed cycle.
\end{obs}
Hence, in any feasible solution the $2m$ vertices $r_1, ..., r_m$ and $\ell_1, ..., \ell_m$ are covered.
So, the only remaining uncovered vertices are $t_j, s_j, p_j$ for each $j$ and the element vertices, $x_i^1,...,x_i^{k+1}$ for each $i$.

Notice that if $t_j\ell_j$ is oriented as $\ora{t_j\ell_j}$ then no vertex $x_i^h$ can
be covered for any $x_i\in S_j$, $1\le h\le k+1$ without forming a cycle.
However if $\ell_jt_j$ is oriented as $\ora{\ell_jt_j}$ then each of those element vertices $x_i^h$ can be covered by $\ora{x_i^hs_j}$ and so we ``equate'' this case with choosing the set $S_j$ as part of a solution to the set cover problem.
Thus minimizing the number of subsets $S_j$ whose union is $X$ will be seen
to be equivalent to maximizing the number of $t_j$'s that are not covered.
This is the basic idea of our proof. 

\subsubsection*{The Proof}

  We show that there is a solution to the constructed instance of {\sc
    tree augmentation} that covers at least $(k+1)n+4m-k$ vertices if and only
  if there is a solution to the instance of {\sc set cover}.

\paragraph{Solutions to {\sc set cover} imply solutions to {\sc tree
  augmentation}:}
  
Suppose there is a solution ${\cal C}'\subseteq {\cal C}$ to the {\sc
  set cover} instance, that is, $|{\cal C}'|\le k$ and $\cup_{S_j\in
  {\cal C}'} S_j = X$.  Note that if $|{\cal C}'|< k$ we can always
add subsets to the solution and it will still cover all the elements.
So, without loss of generality, we will assume that $|{\cal C}'|=k$.
We orient a subset $E' \subset E$ so that it covers $(k+1)n+4m-k$ vertices. Observation \ref{obs:complete} then says that the orientation of $E'$ can
be extended to an orientation of all of $E$ that covers at least $(k+1)n+4m-k$ vertices and
hence is a solution to the {\sc tree augmentation} instance.

\begin{description}
\item[$E'_1$: Forced edges] 
	As discussed earlier, the edges $\ell_ju$ and $r_ju$ are back edges and so must be oriented toward the root: $\overrightarrow{\ell_ju}, \overrightarrow{r_ju}$.  	These arcs then cover the $2m$ vertices $\ell_j$ and $r_j$ for
	$1\le j \le m$.

\item[$E'_2$: $S_j\in {\cal C}'$] 
	For each $S_j\in {\cal C}'$, orient $\ell_jt_j$ and $p_jr_j$ as $\ora{\ell_jt_j}$, $\ora{p_jr_j}$ respectively. These arcs then cover the $k$ $p_j$'s where $S_j\in {\cal C}$.
	Also for each $x_i \in S_j$, orient the edges $\{s_jx_i^1,\ldots, s_jx_i^{k+1}\}$ toward $s_j$.  
	Since ${\cal C}'$ covers the elements, each $x_i^h$ will have at least one non-tree edge oriented away from it, and hence this covers all of the $(k+1)n$, $x_i^j$ vertices.
	Therefore the orientation of the edges of $E'_2$ cover $(k+1)n + k$ additional vertices. 	
	Note that adding the edges of $E'_2$ to $\overrightarrow{T} \cup E'_1$ does not introduce any cycles because it only directs edges {\em from} the element paths {\em to} the set paths.
	
\item[$E'_3$: $S_j \not \in {\cal C}'$] 
	For each $S_j \not \in {\cal C}'$, we orient the edge $t_j\ell_j$ as $\ora{t_j\ell_j}$ and orient all the edges $\{s_jx_i^1,\ldots, s_jx_i^{k+1}\}$ away from $s_j$.  
	The edges of $E'_2$ cover the $2(m-k)$ additional vertices $t_j$ and $s_j$ where $S_j\not \in {\cal C}$. 
	Notice the solution is still acyclic, since for each $S_j \not \in {\cal C}'$, there is no arc directed into any vertex in the set $\{r_j,s_j,t_j\}$.
	
\end{description}
It follows $E' = E'_1 \cup E'_2 \cup E'_3$ is a feasible solution of size $(k+1)n + 4m - k$ as desired.  

\paragraph{No solution to {\sc set cover} implies no solution to
  {\sc tree augmentation}:}
Suppose there is no subcollection of $\cal C$ of size at most $k$ that
covers all the elements of $X$.  We will show that any feasible orientation of
the non-tree edges
$\overrightarrow E$ will cover less than $(k+1)n+4m-k$ vertices.

We have two facts resulting from the existence of the set paths to the
root in $\overrightarrow T$ and directed paths from $\ell_j$ through
all the element vertices and through the collection path to the root.
Since we assume all the non-tree edges are oriented (w.l.o.g.\ by
Observation~\ref{obs:complete}), the following facts must hold or else
there would be a directed cycle.
\begin{enumerate}[{Fact} 1]
\item If $s_j$ is covered, then $p_j$ is not covered. \label{fact1}

\item If $\overrightarrow{x_i^hs_j} \in \overrightarrow E$ for some
  $h$ ($1\le h\le k+1$), then $\overrightarrow{\ell_jt_j} \in
  \overrightarrow E$ and $t_j$ cannot be covered. \label{fact2}
\end{enumerate}

We have two cases: either some representative vertex $x_i^h$ of each 
element $x_i$ is
covered or there is some element $x_i$ all of whose representatives
$x_i^1, x_i^2,\ldots, x_i^{k+1}$ are not covered.

We start with the former case.  Suppose that for each $i\in [1, \ldots
,n]$ there is some $h(i)\in [1,\ldots ,k+1]$ such that $x_i^{h(i)}$ is
covered.  Let $J$ be the subset of indices $[1,\ldots,m]$ such that
the $x_i^{h(i)}$'s are covered by edges oriented to $s_j$ for $j \in
J$.  As we assume that the best set cover requires more than $k$ sets,
$|J| > k$.  By Fact~\ref{fact2}, $t_j$ is not covered for any $j \in
J$.  Then, by Fact~\ref{fact1}, it follows that for $j\in J$ at most 3
of $\{\ell_j, r_j, s_j, t_j, p_j\}$ are covered.  Also note by Facts~1 and~2, at most $4$ of $\{\ell_j, r_j, s_j, t_j, p_j\}$ can be
covered for $j\not\in J$.  Therefore, $\overrightarrow E$ can cover at
most $4m-|J|+(k+1)n < 4m-k+(k+1)n$ vertices.

We now prove the latter case.  Suppose that for some $i \in
[1,\ldots,n]$, $x_i^h$ is not covered for any $h \in [1,\ldots, k+1]$.
By Fact~\ref{fact1}, it follows that at most 4 of $\{\ell_j, r_j, s_j,
t_j, p_j\}$ are covered.  Therefore $\overrightarrow E$ can cover at
most $4m+(k+1)(n-1) = 4m+(k+1)n-k-1$ vertices.

\section{Bounded treewidth {\sc tree augmentation}}
\label{sec:boundedtw}

We show in this section that for $G$ of bounded treewidth, {\sc tree
augmentation} can be solved optimally via dynamic programming.  {\em
Treewidth} is a measure of how far a graph is from being a tree.  It
is known that many NP-hard graph problems become tractable
in graphs of bounded treewidth~\cite{CM92}.
Formally, $G = (V,E)$ has treewidth
$\tw$ if there is a tree decomposition $(\Upsilon, \Gamma)$ of $G$ where
each node $\nu \in \Upsilon$ corresponds to
a subset $S_{\nu}$ of $V$,
and $\Gamma$ is a tree on $\Upsilon$, satisfying
the following four properties.
\begin{enumerate}
\item No more than $\tw+1$ vertices of $V$ are mapped to any one node
  of $\Gamma$, i.e. $|S_{\nu}|\le \tw+1$ for $\nu\in \Upsilon$.
\item 
 The union of $S_{\nu}$ is  $V$, i.e. every
vertex of $V$ is mapped to some node of $\Gamma$.
\item 
 For every edge $uv$ in $E$, $u$ and $v$ are mapped to some
  common node of $\Gamma$.
\item 
If, for $\nu_1\in \Upsilon$ and $\nu_2\in \Upsilon$, $S_{\nu_1}$ and
$S_{\nu_2}$ contain a common vertex of $V$, then for all nodes $\nu$
of in the (unique) path between $\nu_1$ and $\nu_2$ in $\Gamma$,
$S_{\nu}$ contains $v$ as well.
\end{enumerate}
We may assume without loss of generality that $\Gamma$ is a rooted
binary tree with $O(|V|)$ nodes~\cite{Bodlaender93}.  
Note that given a graph $G$ and an integer $k$ it is NP-hard to determine if the treewidth of $G$ is at most $k$~\cite{ACP87}.
However, for any fixed constant $k$, Bodlaender gives a linear time algorithm which determines if a graph $G$ has treewidth at most $k$, and if so, finds a tree-decomposition of $G$ with treewidth at most $k$ in linear time~\cite{Bod96}.

In this section, we prove:
\begin{theorem}\label{thm:tw}
For any constant $k$, and any graph $G$ of treewidth at most $k$,
{\sc Tree augmentation} 
can be solved in linear time using dynamic programming.
\end{theorem}

\begin{proof}
Let $(\Upsilon, \Gamma)$  be a tree decomposition of $G$ of treewidth at most $k$.
Since $\Gamma$ is rooted, for each node $\nu \in \Gamma$, we can define $\Gamma_\nu$ as the subtree of $\Gamma$ rooted at $\nu$.
For a tree node $\nu$ of $\Gamma$, let $G_{\nu}$ be the subgraph of $G$ whose vertex set is $\{v \in G~|~ v \in S_\nu~s.t.~\nu \in \Gamma_\nu\}$ and edge set is $\{uv \in E(G)~|~ u,v \in S_\nu~s.t~\nu \in \Gamma_\nu\}$.  

For each tree node $\nu$ of $\Gamma$, each permutation $P$ of the vertices of $S_\nu$ 
and each subset $C$ of $S_\nu$, we determine two tables, $\ora {\mbox{\sc Tab}}_\nu[P,C]$ and $ {\mbox {\sc Tab}}_\nu[P,C]$:
\begin{enumerate}
\item [(i)] $\ora {\mbox{\sc Tab}}_\nu[P,C]$ is an optimal orientation of the edges $G_{\nu}$
that is consistent with permutation $P$ and covers exactly the vertices $C$, such that the number of vertices covered by $\ora {\mbox{\sc Tab}}_\nu[P,C]$ is maximized.
\item [(ii)] ${\mbox {\sc Tab}}_\nu[P,C]$ is the number of vertices covered by $\ora {\mbox{\sc Tab}}_\nu[P,C]$
\end{enumerate}
If a permutation $P$ contradicts the partial order on the vertices enforced by $\ora T$, then we set {\sc Tab}$_\nu[P,C]=-\infty$ and $\ora {\mbox{\sc Tab}}_\nu[P,C] = \emptyset$ for all $C\subseteq S_{\nu}$.  Likewise, if no orientation of $G_\nu$ can cover $C$, we set {\sc Tab}$_\nu[P,C]=-\infty$ and $\ora {\mbox{\sc Tab}}_\nu[P,C] = \emptyset$ for all permutations $P$ of $S_\nu$.  For simplicity of presentation, we assume that all entries of {\sc Tab} are initialized to $-\infty$ and all entries of $\ora{\mbox{\sc Tab}}$ are initialized to $\emptyset$.

For the root $r$ of $\Gamma$, it follows that the maximum of ${\mbox {\sc Tab}}_r[P,C]$ taken over all permutations $P$ and subsets $C$ of $S_r$ is the value of an optimal solution and the corresponding $\ora {\mbox{\sc Tab}}_\nu[P,C]$ is an optimal solution.
Hence, to complete the proof of Theorem \ref{thm:tw}, it is enough to show how to determine the entries of our dynamic programming table.
We do so in two steps.
First, we determine the entries of {\sc Tab}$_\nu$ for each leaf $\nu \in \Gamma$.
Second, we determine the entries of {\sc Tab}$_\nu$ for each non-leaf node $\nu \in \Gamma$ given that the entries of {\sc Tab}$_{\nu_1}$ and {\sc Tab}$_{\nu_2}$ for its associated child nodes, $\nu_1$ and $\nu_2$, have already been determined.  

Assume  $\nu \in \Gamma$ is a leaf node.  For each permutation $P$ that is consistent with $\ora{T}$, $\ora{E_\nu}$ be the orientation of the edges of $G_\nu$ implied by $P$ and let $C^\star$ be the subset of $S_\nu$ covered by $\ora{E_\nu}$.  By construction,  {\sc Tab}$_\nu$ contains the desired entries.

Assume $\nu$ is an internal node such that the {\sc Tab} entries associated with the children of $\nu$, namely $\nu_1$ and $\nu_2$, have been computed.  
We consider all permutations $P_1$ of $S_{\nu_1}$ and $P_2$ of $S_{\nu_2}$ together with all subsets $C_1 \subseteq S_{\nu_1}$ and $C_2 \subseteq S_{\nu_2}$.
We use the entries of ${\mbox{\sc Tab}}_{\nu_1}[P_1,C_1]$ and ${\mbox{\sc Tab}}_{\nu_2}[P_2,C_2]$ to construct ${\mbox{\sc Tab}}_{\nu}[P,C]$.
Now, not all choices of $P_1, P_2, C_1$ and $C_2$ lead to valid solutions.
Indeed, the partial order given by the permutation $P_1$ must be consistent with $P$; analogously $P_2$ must be consistent with $P$. 
Additionally, if $v$ is in $C_1$ and $v \in S_\nu$ then $v$ must also be $C$; the analogous condition holds for $C_2$. 
For $i \in \{1,2\}$, we call $P_i$ and  $C_i$ {\em good} if they satisfy these conditions.
We now show that 
\begin{eqnarray}
{\mbox{\sc Tab}}_\nu[P,C] 
= |C| &+& \max_{{\rm good}~P_1, C_1}
\left\{  
{\mbox{\sc Tab}}_{\nu_1}[P_1,C_1] - |C_1 \cap C|
\right\} 
\\\nonumber &+& \max_{
{{\rm good}}~P_2, C_2}
\left\{  
{\mbox{\sc Tab}}_{\nu_2}[P_2,C_2]  - |C_2 \cap C|
\right\},
\label{eqn1}
\end{eqnarray}
and $\ora {\mbox{\sc Tab}}_\nu[P,C]$ can be determined from the entries of {\sc Tab}$_{\nu_1}$ and {\sc Tab}$_{\nu_2}$.

Fix a permutation $P$ and subset $C$ of $S_\nu$.
Let $\ora O$ be any orientation of $G_\nu$ that is valid with respect to $P$ and $C$.
Let $A_1$ be the vertices covered in $G_{\nu_1}$ and $A_2$ be the vertices covered in $G_{\nu_2}$.
By the fourth property of tree decompositions, the number of vertices covered by $\ora O$ is exactly $|C| + |A_1 \setminus C| + |A_2 \setminus C|$.
Letting $\ora O_1$ be the restriction of $\ora O$ to $G_{\nu_1}$, for any good $P_1$ and $C_1$ we have $|A_1| \le {\mbox{\sc Tab}}_{\nu_1}[P_1,C_1].$
Hence, $|A_1| \le \max_{{\rm good}~ P_1, C_1}{\mbox{\sc Tab}}_{\nu_1}[P_1,C_1],$ and so, 
$|A_1 \setminus C| \le \max_{{\rm good}~ P_1, C_1}( {\mbox{\sc Tab}}_{\nu_1}[P_1,C_1] - |C_1 \cap C|).$
Similarly, $|A_2 \setminus C| \le \max_{{\rm good}~ P_2, C_2}( {\mbox{\sc Tab}}_{\nu_2}[P_2,C_2] - |C_2 \cap C|).$
Hence, 
${\mbox{\sc Tab}}_\nu[P,C]$ is at most the righthand side of Equation \ref{eqn1}.

To complete the proof it is enough to show there exists good $P_1, C_1$ and good $P_2, C_2$ such that the corresponding entries of $\ora {\mbox{\sc Tab}}_{\nu_1}[P_1,C_1]$ and $\ora {\mbox{\sc Tab}}_{\nu_2}[P_2,C_2]$ can be used to give an acyclic orientation of value equal to the righthand side of Equation \ref{eqn1}.
Let good $P_1,C_1$ be chosen to maximize ${\mbox{\sc Tab}}_{\nu_1}[P_1,C_1] - |C_1 \cap C|$, and let good $P_2,C_2$ be chosen to maximize ${\mbox{\sc Tab}}_{\nu_2}[P_2,C_2] - |C_2 \cap C|$.
We first orient the edges $S_\nu$ by the permutation $P$.  
By the fourth property of tree decompositions, every remaining edge is completely contained in either $G_{\nu_1}$ or completed contained in $G_{\nu_2}$.
For these edges we use the orientations of $\ora {\mbox{\sc Tab}}_{\nu_1}[P_1,C_1]$ and $\ora {\mbox{\sc Tab}}_{\nu_2}[P_2,C_2]$, respectively.  
Clearly, this orientation covers the desired number of nodes; it only remains to show it is acyclic.  

For the purpose of contradiction, let $u,v \in S_\nu$ be such that the permutation $P$ places $u$ before $v$ but {\sc Path}$_{vu}$, a directed path from $v$ to $u$, already exists. 
If there are multiple such paths, we choose the shortest path.
If {\sc Path}$_{vu}$ contains an internal node $x\in S_{\nu}$, then
{\sc Path}$_{vx}$ and {\sc Path}$_{xu}$ are also directed paths. Since
{\sc Path}$_{vu}$ is the shortest, the permutation $P$ must ensure that $v$ is before $x$ and $x$ is before $u$. However, $P$ also enforces $u$ is before
$v$, which is a contraction.
Therefore, none of the internal nodes in {\sc Path}$_{vu}$ can be in
$S_{\nu}$. By properties 3 and 4 of the tree decomposition, every
internal node of {\sc Path}$_{vu}$ must be contained entirely in one
of the subtrees, say the one rooted at $\nu_1$. In addition, $u$ and
$v$ must be contained in $S_{\nu_1}$.  If a permutation
$P_1$ ensures $v$ is before $u$ then $P_1$ and $P_2$ would be
inconsistent; if $P_1$ ensures $u$ is before $v$ then $\nu_1$ would be a
earlier node in which a cycle appears.
It now follows that $\ora {\mbox{\sc Tab}}_{\nu}[P,C]$ is the desired orientation.
Moreover, since we can compute it by considering the $O((k!)^4)$ possible choices for $P_1,P_2, C_1$ and $C_2$, it follows that it follows that constructing {\sc Tab} can be done in linear time.
\qed 
\end{proof}
The above dynamic programming argument generalizes to the 
weighted case, if we keep track of the node weights rather than 
the number of nodes.
\begin{cor}
For any constant $k$, and any graph $G$ of treewidth at most $k$,
{\sc Weighted Tree augmentation} 
can be solved using dynamic programming.
\end{cor}

\section{A PTAS for {\sc BFS tree augmentation} in planar graphs}
\label{sec:planar}

In this section we consider a special case in which the input graph
$G=(V,\ora T\cup E)$ is planar and $\ora T$ is the breath-first-search
(BFS) tree of $G$. (Specifically, the $T$ is a BFS tree of the
undirected version of $(V,T \cup E)$.)  We show that Baker's technique, described in a moment,
can be used to design a polynomial-time approximation scheme (PTAS)
for this special case.  A PTAS is, for a fixed
integer $d$, a polynomial (in $|G|$) time algorithm that finds a
solution of value achieving at least a $1-{1\over d}$ fraction of
the optimal solution's value.

Baker's technique is a shifting technique for designing PTASes for
planar graph instances of problems~\cite{Baker1994}.  Baker introduced
this technique to solve NP-hard problems such as {\sc independent set}
and {\sc vertex cover}; the constraints for these problems are defined
locally, for neighborhoods of vertices.  Usually, Baker's technique
cannot be used to solve problems with a global constraints, such as
our acyclicity constraint.  However, the technique involves separating
the graph at BFS layers; when $\ora T$ is a BFS tree, we can guarantee
acyclicity across the separators.  We show how to use Baker's
technique here to prove:

\begin{theorem}\label{thm:planar} There is a PTAS for {\sc Tree augmentation} when the input is a planar graph $G = (V, E \cup \ora{T})$ such that $\ora T$ is a BFS tree of $G$.
\end{theorem}

\begin{proof}
Label each vertex with its BFS level from the root of $\ora T$.  Let
$F_i$ be the subset of arcs and edges ($\ora T \cup E$) that have one
endpoint in level $i$ and one endpoint in level $i+1$.  Let $F^k =
\cup_{i = k\bmod d} F_i$ for $k = 1,\ldots d$.

Let ${\cal S}_k$ be the connected components of $(V,(T \cup E)
\setminus F_k)$.  Consider a component $S$ of ${\cal S}_k$ and attach
all the nodes at the smallest level in $S$ to a newly created root
node $r_S$ by arcs to create $S'$; these new arcs plus the arcs of
$\ora T \cap S$ forms a BFS spanning tree $\ora T_S$ directed toward
$r_S$.  Since $\ora T_S$ has depth at most $d$, $S'$ has treewidth at
most $3d$ by way of:
\begin{theorem}[Baker~\cite{Baker1994}]
  Given a planar graph $G$ with rooted spanning tree of depth $d$ a
  tree decomposition of width at most $3d$ can be found in $O(d|G|)$
  time.
\end{theorem}

We are now able to describe the PTAS of Theorem \ref{thm:planar}.
By the algorithm of Section~\ref{sec:boundedtw}, we can optimally solve the {\sc tree augmentation} problem in each of the components of ${\cal S}_k$ in polynomial time. 
For each $k = 1,...,d$, compute the optimal solutions corresponding to each component of ${\cal S}_k$.  
Let $\ora E_k$ be the orientation given by the union of these solutions with the edges of $F^k\cap E$ oriented from high-to-low BFS level.  
Return the best solutions of $\{\ora E_1, \ldots, \ora E_d\}$.
To complete the proof Theorem \ref{thm:planar}, it remains only to prove acyclicity and near-optimality:
\paragraph{Acyclicity} First notice that these orientations are
acyclic.  Any directed cycle in $\ora E_k \cap \ora T$ would have to
include arcs in multiple components of ${\cal S}_k$ and, in
particular, would travel from a low BFS-level to a high
BFS-level and back, crossing $F_i$ in each direction for some $i =
k\bmod d$.  However, since $T$ is a BFS tree, each arc in $\ora T \cap
F_i$ is oriented from level $i$ to level $i+1$.  By design the edges
of $F_i \cap E$ are also oriented from level $i$ to level $i+1$.
Therefore $\ora E_k \cap \ora T$ is acyclic.
\paragraph{Near-Optimality} Let $E^\star$ be a minimal subset of $E$ such that an optimal solution
$\ora E^\star$ of $G=(V,\ora T\cup E^\star)$ covers as many vertices
as an optimal solution for $G=(V,\ora T\cup E)$.  That is, every
vertex is the starting point for at most one arc of $\ora E^\star$ and
so the maximum number of vertices that can be covered is $|E^\star|$.

$F^k$ is a partition of a subset of $\ora T \cup E$.  Therefore
\begin{equation}\label{ineq}
\min_k |E^\star \cap F^k| \leq {1\over d} \sum_k |E^\star \cap F^k|
\leq {1\over d} |E^\star|
\end{equation}
Consider the index $k^\star$ that is the argument the above minimum.
For each component $S$ of ${\cal S}_k$, $\ora E_{k^\star} \cap S$
covers at least as many vertices as $\ora E^\star$, since $\ora
E_{k^\star}$ is optimal for $S$.  Therefore $\ora E_{k^\star}$ covers
at least $|E^\star|-|E^\star \cap F^{k^\star}|$ vertices.  By
Inequality~\ref{ineq}, $\ora E_{k^\star}$ covers at least $(1-{1\over
  d})|E^\star|$ vertices.
This completes the proof.
% of Theorem \ref{thm:planar}.
\qed \end{proof}
\begin{cor}
There is a PTAS for {\sc Weighted Tree augmentation} when the input is
a planar graph $G = (V, E \cup \ora{T})$ such that $\ora T$ is a BFS
tree of $G$.
\end{cor}

\section{Two-Arm {\sc tree augmentation}}
\label{sec:twoarm}
We consider a special case in which the tree $\ora T$ consists of
exactly two directed paths and give a polynomial-time dynamic program
for finding an optimal solution to {\sc tree augmentation}.  $\ora T$
has root $\ell_0 = r_0$ and two directed paths to the root: the {\em
  left arm} with vertices in order from leaf to root $\ell_{n_\ell},
\ell_{n_\ell - 1}, \ldots, \ell_{1}, \ell_0$ and the {\em right arm}
with vertices in order from leaf to root $r_{n_r}, r_{n_r - 1},
\ldots, r_{1}, r_0$.  We use inequalities to compare the indices of
these vertices (i.e., $r_i < r_j$ if $i < j$).

Any edge in $E$ with both endpoints in a single
arm is a back edge and must be oriented toward the root; we denote
this orientation by $\ora B$ as in Section~\ref{sec:comp}.  Each cross
edge $e \in C \subseteq E$ has a left endpoint $\ell(e)$ and a right
endpoint $r(e)$.  We let $\ora e = \ora{\ell(e)r(e)}$ denote the
left-to-right orientation and $\ola e = \ola{\ell(e)r(e)}$ denote the
right-to-left orientation of $e$.

We sort the cross edges first by left, then by right endpoint.  Namely
$C = \{e_1, e_2, \ldots, e_m\}$ such that $i < j$ only if either
$\ell(e_i) < \ell(e_j)$ or $\ell(e_i)= \ell(e_j)$ and $r(e_i) <
r(e_j)$.  (Note that we may assume that there are no parallel edges as
all parallel edges would need to be oriented consistently to maintain
acyclicity.)  For each $k = 0,1,
\ldots, m$ and each $j = 1,\ldots,n_r+1$ we determine
\[
  \mbox{the orientation }\ora{C_{j,k}}\mbox{ of }C_{j,k} =
  \{e_i \ : \ i\le k, r(e_i) < r_j \} \mbox{ that maximizes}
\]\vspace{-5mm}
\[\mbox{the number } c_{j,k} \mbox{ of
    endpoints of }C_{j,k}\mbox{ covered by }\ora{C_{j,k}}\cup\ora B.
\]
(The notation $\ora{C_{j,k}}$ does not indicate orienting all the edges in
$C_{j,k}$ from left to right.)  
Clearly, the solution to {\sc tree augmentation} is $\ora{C_{n_r,m}}
\cup \ora B$.

Note that the sets $C_{j,0}$ are empty and so the values $c_{j,0}$
denote the number of vertices covered by $\ora B$.  We denote the
coverage by $\ora B$ by:
\[
 \delta(v) = \left\{
   \begin{array}{ll}
     1 & \mbox{if }v\mbox{ is not covered by }\ora B \\
     0 & \mbox{otherwise}
   \end{array}
 \right.
\]
For any $j$ and for $k > 0$, we determine $\ora{C_{j,k}}$ from
$\ora{C_{j,t}}$ for $t < k$ and $\ora{C_{t,k}}$ for $t < k$.  This
allows us to compute $\ora{C_{n_r,m}}$ via dynamic programming.

If $r_j \le r(e_k)$, then $e_k \notin C_{j,k}$ and so $C_{j,k} =
C_{j,k-1}$; therefore $\ora{C_{j,k}} = \ora{C_{j,k-1}}$.  If $r_j >
r(e_k)$, then $e_k \in C_{j,k}$ and we take the better of two options:
$\ora{e_k}$ or $\ola{e_k}$.  
Refer to Figure~\ref{fig:2arm}.

\begin{figure}[ht]
%  \centering
\centerline{
\input{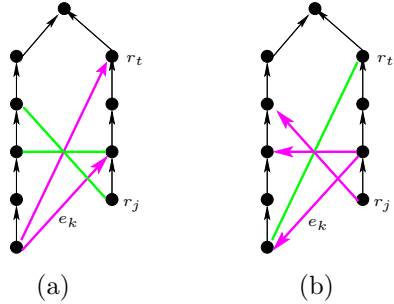}
}
  \caption{The case in which $e_k$ is oriented (a) $\protect\ora{e_k}$, from left to right and (b) $\protect\ola{e_k}$, from right to left.}
  \label{fig:2arm}
\end{figure}

In the $\ora{e_k}$ option, to ensure acyclicity, any edge in $C_{j,k}$ that
shares $e_k$'s endpoint must also be oriented from left to right.  
We define $t$ such that $e_t$ is the last edge of $C_{j,k}$ with $\ell(e_t) <
    \ell(e_k)$.
Any
acyclic orientation of $C_{j,t}$ combined with this will also be
acyclic as any introduced cycle would go through $\ell(e_k)$
which we have prevented.  The only additionally covered vertex is
$\ell(e_k)$ if it is not already covered by $\ora B$.
Formally, letting $e_t$ be the last edge of $C_{j,k}$ with $\ell(e_t) <
    \ell(e_k)$:
\[
\begin{array}{ll}
 \mbox{Option }\ora{e_k}: &
  \left\{
    \begin{array}{lll} 
      \ora{C_{j,k}} &=&
      \ora{C_{j,t}} \cup \{ \ora{e}\ : \ \ell(e) = \ell(e_k),\ e \in C_{j,k} \} \\
      c_{j,k} &=& c_{j,t} + \delta(\ell(e_k))
    \end{array}
  \right.
\end{array}
\]

In the $\ola{e_k}$ option, to ensure acyclicity, any edge in $C_{j,k}$
with right endpoint after $r(e_k)$ must also be oriented from right to
left. We define $t$ such that $r_t = r(e_k)$. Note that $t<j$ since
we are considering the case $r_j>r(e_k)$.
Any acyclic orientation of $C_{t,k}$ combined with this will
also be acyclic as any introduced cycle would have to go through
$r(e_k)$, which we have prevented.  The right endpoints of these newly
oriented edges may become covered if they were not already covered by
$\ora B$.  Formally, letting $r_t = r(e_k)$:
\[
\begin{array}{ll}
  \mbox{Option }\ola{e_k}: &
  \left\{
    \begin{array}{lll} 
      \ora{C_{j,k}} &=&
      \ora{C_{t,k}} \cup \{ \ola{e}\ : \ r(e) \ge r_t,\ e \in
      C_{j,k} \} \\
      c_{j,k} &=& c_{t,k}+\sum_{\ola{e}:r(e) \ge r_t, e \in
      C_{j,k}} \delta(r(e))
    \end{array}
  \right.
\end{array}
\]
It is not difficult to see that by storing these two options for
each value of $j$ and $k$, one can give a polynomial-time implementation
of the dynamic program.
%\marginpar{check this; not sure about linear time. Can this set $\{ r(e) \ge r_t,\ e \in
%      C_{j,k} \}$ be found in constant time? Also the table size is not linear.}

\begin{theorem}
{\sc Tree augmentation} in the special case of 2 arms can be solved
optimally using dynamic programming in polynomial time.
\end{theorem}
It is easy to see the following generalization to the weighted case.
In the dynamic programming, instead of having binary $\delta(v)$ we have
$\delta(v)$ reflect the weight of node $v$.  
\begin{cor}
{\sc Weighted Tree augmentation} in the special case of 2 arms can be
solved optimally using dynamic programming in polynomial time.
\end{cor}

\subsection*{Conclusion}
In this paper we study improving robustness in next-hop routing by
modeling it as a graph theoretic problem {\sc tree augmentation}.
This work leads to a number of open problems. For example, can the
dynamic programming approach be applied to more special cases?  We
note that the special case of multiple arms is not immediately
amenable to dynamic programming. More generally, what is the
complexity when the problem has a bounded number of leaves?  Does the
problem admit a better-than $\frac{1}{2}$-approximation in the general case?

\bibliographystyle{plain} 
\bibliography{biblio}

\end{document}